\documentclass[11pt, letter]{article}
\usepackage{amsfonts}
\usepackage{amsmath}
\usepackage{epsfig}
\usepackage{mathrsfs}
\usepackage{graphicx}
\usepackage{amssymb}
\usepackage[scriptsize]{subfigure}

\newtheorem{theorem}{Theorem}

\newtheorem{definition}{Definition}
\newtheorem{lemma}{Lemma}

\newenvironment{proof}[1][Proof]{\noindent \textbf{#1.} }{\qedsymbol}
\newcommand{\qedsymbol}{\hspace{\fill}\rule{1.5ex}{1.5ex}}
\input{tcilatex}
\def\baselinestretch{1.30}
\def\b0{\mbox{\boldmath $0$}}

\def\bee{\mbox{\boldmath $e$}}

\def\by{\mbox{\boldmath $y$}}

\def\buno{\mbox{\boldmath $1$}}

\begin{document}

\title{\vspace{-1.7cm}{\Huge Distributed Consensus
over Wireless Sensor Networks Affected by Multipath
Fading\vspace{-0.2cm}}}
\author{Gesualdo Scutari and Sergio Barbarossa\\
{Dpt. INFOCOM, Univ. of Rome \textquotedblleft La Sapienza
\textquotedblright, Via Eudossiana 18, 00184 Rome, Italy}\\
{ E-mail: \texttt{$\{$scutari, sergio$\}$@infocom.uniroma1.it}}.
\thanks{This work has been partially funded by the WINSOC project, a Specific
Targeted Research Project (Contract Number 0033914) co-funded by the INFSO
DG of the European Commission within the RTD activities of the Thematic
Priority Information Society Technologies, and by ARL/ERO Contract
N62558-05-P-0458.}\\\\}
\date{\small Paper submitted to \textit{IEEE Transactions on Signal Processing}, August 2007.\\ Revised November 30, 2007. Accepted January 14, 2008.}
\maketitle
\vspace{-0.8cm}
\begin{abstract}
The design of sensor networks capable of reaching a consensus on a
globally optimal decision test, without the need for a fusion
center, is a problem that has received considerable attention in the
last years. Many consensus algorithms have been proposed, with
convergence conditions depending on the graph describing the
interaction among the nodes. In most works, the graph is undirected
and there are no propagation delays. Only recently, the analysis has
been extended to consensus algorithms incorporating propagation
delays. In this work, we propose a consensus algorithm able to
converge to a {\it globally optimal} decision statistic, using a
\emph{wideband} wireless network, governed by a fairly simple MAC
mechanism, where each link is a multipath, frequency-selective,
channel. The main contribution of the paper is to derive necessary
and sufficient conditions on the network topology and sufficient
conditions on the channel transfer functions guaranteeing the
exponential convergence of the consensus algorithm to a globally
optimal decision value, for \emph{any} bounded delay condition.
\end{abstract}

\section{Introduction}

Distributed algorithms for achieving consensus in wireless sensor
networks, without the need for a fusion center, have been the
subject of many recent works. Two excellent tutorials on the subject are \cite%
{Olfati-Saber-Murray-ProcIEEE07, Ren-Beard-Control-Magazine} (see
also references therein). The conditions for achieving a consensus
over a globally optimal decision test ultimately depend on the
properties of the graph modeling the interaction among the nodes. Most works consider
undirected graphs and neglect propagation delays. There are only a
few works that study the impact of delays in consensus-achieving
algorithms, namely \cite{Olfati-Saber}$-$\cite{Lee-Spong-06},
focusing on time-continuous systems, and
\cite{Tsitsiklis-Bertsekas-Athans}$-$\cite{Blondel-Tsitsiklis},
dealing with discrete-time systems. Among these works, it is useful
to distinguish between consensus algorithms, \cite%
{Olfati-Saber-Murray-ProcIEEE07}$-$\cite{Olfati-Saber}, where the states of
all the sensors converge to a prescribed function (typically the average) of
the sensors' initial values, and agreement algorithms, \cite{Strogatz}$-$%
\cite{Blondel-Tsitsiklis}, typically used for coordinating the
motion of sets of vehicles, where the states of the nodes converge
to a common value, but this value is not a specified function of the
initial values.  A recent work proposed a randomized gossip algorithm
\cite{boyd-gosh-prabhakar-shah} to achieve distributed consensus,
with a simple interaction mechanism, where each node interacts with
one node at the time, in a randomized fashion.

 In this work, we are interested in distributed
consensus algorithms where the consensus coincides with a globally
optimal decision statistic. Our goal is to derive the conditions on
the channels between each pair of nodes, guaranteeing that each
sensor will eventually converge to the globally optimal decision
statistic, in a totally distributed manner, i.e. without requiring
the presence of a fusion center. In \cite{Olfati-Saber-Murray-ProcIEEE07, Olfati-Saber}, the authors
provided necessary and sufficient conditions for the convergence of
a linear consensus protocol, in the case of a common time-invariant
delay value for all the links, i.e., $\tau _{ij}=\tau $ $\forall i\neq j$, and assuming \emph{%
symmetric} channels among the nodes (modeled as a undirected graph).
Under these assumptions, the average consensus in
\cite{Olfati-Saber-Murray-ProcIEEE07, Olfati-Saber} is reached if
and only if the common delay $\tau $ is smaller than a
topology-dependent value. However, the assumptions of homogeneous
delays and nonreciprocal channels are not appropriate for describing
the propagation in a common network deploying scenario, where the
delays depend on traveled distances and the communication channels
may be \emph{asymmetric}. In
\cite{Scutari-Barbarossa-Delay-SPAWC-07}, we generalized the
consensus algorithms to networks with inhomogeneous delays and
asymmetric flat-fading channels. In this correspondence, we extend
our previous work to the more general case where each link is
modeled as a multipath channel. We assume baseband communications,
motivated by the use of impulse radio technologies. The main
contributions of this paper are the following: i) We provide
necessary and sufficient conditions on the network topology and
sufficient conditions on the transfer function of each channel
ensuring global convergence to the optimal decision test, for
\textit{any} set of finite propagation delays; ii) We prove that the
convergence is exponential, with convergence rate depending, in
general, on the channel parameters and propagation delays;
iii) We show how to reach a distributed consensus coinciding with
the globally optimal decision statistics, achievable by a
centralized system having error-free access to all the nodes
measurements and observation parameters, without the need of
estimating neither the channel coefficients nor the delays.

\section{How to Achieve Consensus on a Globally Optimal Decision Test in a
Decentralized Way\label{System-model}}

Let us consider a set of $N$ sensors, each measuring a scalar parameter $y_i$%
, $i=1, \ldots, N$. The goal of the network is to compute a sufficient
statistic of the measured data expressible as
\begin{equation}
f(y_{1},y_{2},\ldots ,y_{N})=h\left[ \frac{\dsum_{i=1}^{N}c_{i}g_{i}(y_{i})}{%
\dsum_{i=1}^{N}c_{i}}\right], \vspace{-0.2cm}  \label{f}
\end{equation}
where $\left\{ c_{i}\right\} $ are positive coefficients and $\{g_{i}\}$ and
$h$ are arbitrary (possibly nonlinear) real functions on $\mathbb{R}$, i.e.,
$g_{i},h:%
\mathbb{R}
\mapsto
\mathbb{R}
$. Even though the class of functions expressible as in (\ref{f}) is
not the most general one, it does include many cases of practical
interest, like, e.g., best linear unbiased estimation or ML
estimation under linear signal models, multiple hypothesis testing,
detection of Gaussian processes in Gaussian noise, computation of
maximum, minimum, geometric mean or the histograms of the gathered
data \cite{Scutari-Barbarossa-Delay-SPAWC-07,
Barbarossa-Scutari-Magazine}. In this
paper, we consider only the scalar observation case, but the extension of (%
\ref{f}) to the vector case is straightforward, along the same guidelines of
\cite{Scutari-Barbarossa-Delay-SPAWC-07, Barbarossa-Scutari-Journal}.

To compute functions in the form (\ref{f}) in a distributed way, we
consider a linear interaction model among the nodes, and we
generalize the approach of \cite{Scutari-Barbarossa-Delay-SPAWC-07,
Barbarossa-Scutari-Journal} to a network where the channel between
each pair of nodes is a \emph{multipath} channel, with, in general,
asymmetric channel coefficients and geometry-dependent delays.
In each node there is a dynamical system whose state $x_{i}(t; \by)$
evolves according to the following linear differential equation
\begin{equation}
\begin{array}{l}
\dot{{x}}_{i}(t; \by)=g_{i}(y_{i})+\dfrac{K}{c_{i}}\dsum\limits_{j\in \mathcal{N}%
_{i}}\dsum\limits_{l=1}^{L}a_{ij}^{(l)}\,\left( x_{j}(t-\tau
_{ij}^{(l)}; \by)-x_{i}(t; \by)\right) ,\quad t>0, \\
x_{i}(\vartheta; \by)=\widetilde{\phi }_{i}(\vartheta ),\quad
\vartheta \in
\lbrack -\tau ,\ 0],%
\end{array}%
\quad i=1,\ldots ,N,  \label{linear delayed system}
\end{equation}
where $\mathbf{y}=\{y_i\}_{i=1}^{N}$ is the set of measurements; $g_{i}(y_{i})$ is a function of the local measurement, whose
form depends on the specific decision test; $c_{i}$ is a positive
coefficient that is chosen in order to achieve the desired
consensus, as in (\ref{f}); $K$ is a positive coefficient
controlling the convergence rate; $a_{ij}^{(l)}$ and $\tau
_{ij}^{(l)}$ are the amplitude and the delay associated to the
$l$-th path of the channel between nodes $i$ and $j$;
$\mathcal{N}_{i}=\{j=1,\ldots ,N: \exists \, a_{ij}^{(l)}\neq 0,\,
l=1, \ldots, L\}$ denotes the set of neighbors of node $i$, i.e.,
the nodes that send signals to node $i$. It is worth noticing that
the state function of, let us say, node $i$ depends, directly, only
on the measurement $y_i$ taken by the node itself and only
indirectly on the measurements gathered by the other nodes. In other
words, even though the state $x_i(t; \by)$ gets to depend,
eventually, on all the measurements, through the interaction with
the other nodes, each node needs to know only its own measurement.

The channel through which node $r$ receives the signal from node $q$
is a multipath channel with transfer function $H_{rq}(j\omega
)=\sum\nolimits_{l=1}^{L}a_{rq}^{(l)}e^{-j\omega \tau _{rq}^{(l)}},$
for all $r\neq q.$ We assume that the channel coefficients are
sufficiently slowly varying to be considered constant for the time
interval necessary for the network to converge, within a prescribed
accuracy. In Section 3, we will show that the convergence of
(\ref{linear delayed system}) is exponential and we will derive a
bound for the convergence rate. Knowing this rate, our method is
applicable for those channels whose coherence time is sufficiently
greater than the convergence time. We are interested in baseband
communications, motivated from a possible implementation of the
radio interface allowing for the interaction described by
(\ref{linear delayed system}) with an impulse radio using
pulse-position modulation (IR-PPM), where the position of the pulse
transmitted by node $i$ is proportional to the state of node
$x_i(t)$.  In general, we allow the channels to be asymmetric, i.e.,
$a_{rq}^{(l)}$ may be different from $a_{qr}^{(l)}$ (and thus
$H_{rq}(j\omega )\neq H_{qr}(j\omega )$)$.$ We also assume,
realistically, that the maximum delay is bounded, with maximum value
$\tau =\max_{r,q,l}\tau _{rq}^{(l)}.$ Because of the delays, the state evolution (\ref%
{linear delayed system}) for, let us say, $t>0$, is uniquely defined
provided that the initial state variables $x_{i}(t; \by)$ are
specified in the
interval from $-\tau $ to $0,$ i.e., $x_{i}(\vartheta; \by)=\widetilde{\phi }%
_{i}(\vartheta ),$ for all $i=1,\ldots ,N,$ and $\vartheta \in
\lbrack -\tau ,\ 0].$

Some important comments about the interaction mechanism (\ref{linear
delayed system}) are appropriate. Distributed consensus algorithms
have a clear advantage with respect to centralized systems, as they
are less prone to congestion events or failures of some of the
nodes. They are also inherently scalable. However, as opposed to
centralized systems, they typically require an iterative mechanism
to converge to the desired decision test. In most available works on
distributed consensus, it is tacitly assumed that each node is able
to receive the signals sent by its neighbors separately. This, of
course, requires a proper medium access control (MAC) mechanism to
avoid collisions. But, when combined with the iterative nature of
distributed consensus algorithms, a collision avoidance MAC protocol
may become rather complicated. Even the simple randomized gossip
algorithm of \cite{boyd-gosh-prabhakar-shah} requires some form of
MAC to avoid collisions. Unfortunately, enforcing a MAC control goes
against the requirement of simplicity and scalability, which are
some of the major motivations underlying the use of distributed
consensus algorithms. Conversely, we are interested in distributed
consensus mechanisms where all nodes transmit over a common shared
physical channel and there are no collision avoidance or resolution
mechanisms whatsoever, so that each node receives a linear
combination of the signals transmitted by the other nodes, possibly
through a multipath propagation channel. This motivates the
interaction model expressed by (\ref{linear delayed system}), from
which it turns out that each node does not need to resolve the
received signals to be able to update its own state function. In
this correspondence, we do not study the radio interface allowing
for the node interaction given by (\ref{linear delayed system}).
Nevertheless, some preliminary studies, see e.g.,
\cite{Barbarossa-Scutari-Magazine, Pescosolido-Barbarossa-icassp2008, Pescosolido-radar} suggest that impulse radios with
pulse position modulation or distributed phase-lock circuits are
possible candidates for implementing (\ref{linear delayed system}),
where the state values are exchanged through pulse position
modulation or phase modulation, respectively.

However, the  advantages of distributed consensus based algorithms as described above come at the price of a penalty:  the final consensus 
 is reached through an iterative procedure that consumes
time and energy.   The overall energy necessary
to achieve the final decision is the sum of the powers transmitted
by each sensor multiplied by the convergence time (in the next section, we will give an upper bound of such a value). On one hand, to save energy, we would like
to use the minimum transmit power that ensures network connectivity.
But a small transmit power has an effect on the network
topology, as it leads to a reduced number of links and,
as a consequence, to a small algebraic connectivity.
Hence, a small individual transmit power implies a long convergence
time. Conversely, to reduce the convergence time,
the network should have a high connectivity, but this requires
a large transmit power. It is then intuitive to expect an optimal
trade-off. This trade-off has been  studied in  \cite{Barbarossa-energy}, where we remand to the interested reader.
The focus of this
correspondence is on finding the conditions on the channel
parameters that guarantee the convergence of (\ref{linear delayed
system}) to the desired consensus value.

\bigskip

\noindent \textbf{Consensus on the state derivative}.
Differently from most papers dealing with average consensus problems \cite%
{Olfati-Saber-Murray-ProcIEEE07}$-$\cite{Olfati-Saber}, \cite{Strogatz}$-$%
\cite{Blondel-Tsitsiklis}, we adopt here the alternative definition of
consensus already introduced in our previous works \cite%
{Scutari-Barbarossa-Delay-SPAWC-07}$-$\cite{Barbarossa-Scutari-Magazine}:
We define the consensus (or network synchronization) with respect to
the state \textit{derivative}, rather than to the state.
\begin{definition}
\label{Definition_sync-state}Given the dynamical system in
(\ref{linear delayed system}), we say that a solution
$\{{x}_{i}^{\star }(t; \by)\}$ of (\ref{linear delayed system}) is a
\emph{synchronized state} of the system, if $\dot{{x}}_{i}^{\star
}(t; \by)={\alpha }^{\star }(\by),$ $\forall i=1,2,\ldots ,N$. The
system (\ref{linear delayed system}) is said to \emph{globally}
synchronize if there exists a synchronized state ${\alpha }^{\star
}(\by)$, and \emph{all} the state derivatives asymptotically
converge to this common
value, for \emph{any} given set of initial conditions $\{\widetilde{{\phi }}%
_{i}\},$ i.e., $\lim_{t\mapsto \infty }| \dot{{x}}_{i}(t;
\by)-{\alpha }^{\star
}(\by)| =0,$ $\forall i=1,2,\ldots ,N$, where $\{{x}_{i}(t; \by)\}$ is a solution to (%
\ref{linear delayed system}). The synchronized state is said to be \emph{%
globally asymptotically stable} if the system globally synchronizes,
in the sense specified above.
\end{definition}
\noindent Observe that, according to Definition
\ref{Definition_sync-state}, if there exists a globally
asymptotically stable synchronized state, then it must necessarily
be \emph{unique }(in the derivative).  One of the reasons to
introduce this definition of consensus, as opposed to the consensus on the state \cite%
{Olfati-Saber-Murray-ProcIEEE07}$-$\cite{Olfati-Saber}, \cite{Strogatz}$-$%
\cite{Blondel-Tsitsiklis}, is that, as will be shown in the next
section, the convergence on the state derivative is not affected by
the presence of propagation delays. One more reason is that, in the
presence of coupling noise, state-convergent algorithms give rise to
a noise with diverging variance , whereas the algorithm converging
on the state derivative exhibits a finite variance
\cite{Barbarossa-Scutari-Journal, Barbarossa-Scutari-Magazine}.

\section{Necessary and Sufficient Conditions for Achieving Consensus}

To derive our main results, we rely on some basic notions of
directed graph (digraph) theory, as briefly recalled next. More details are given in \cite[%
Appendix A]{Scutari-Barbarossa-Delay-SPAWC-07}. A digraph $\mathscr{G}$ is
defined as ${\mathscr{G}=}\{{\mathscr{V},\mathscr{E}}\}$, where $\mathscr{V}$
is the set of vertices and $\mathscr{E}\subseteq \mathscr{V}\times %
\mathscr{V}$ is the set of edges, with the convention that $%
e_{ij}=(v_{i},v_{j})\in \mathscr{E}$ if there exists an edge from $v_{j}$ to
$v_{i}$, i.e., the information flows from $v_{j}$ to $v_{i}$. A digraph is
weighted if a positive weight, denoted by $a_{ij}$, is associated with each
edge $e_{ij}$. The in-degree of a vertex is defined as the sum of the
weights of all its incoming edges. The out-degree is similarly defined.
The Laplacian matrix $\mathbf{L=L(%
}\mathscr{G}\mathbf{)}$ of the digraph associated to system
(\ref{linear delayed system}) is $\mathbf{L}=\mathbf{D}-\mathbf{A}$,
where $\mathbf{D}$ is the diagonal matrix of vertex in-degrees and
$\mathbf{A}$ is the adjacency matrix. For reasons that will be
clarified within the proof of next theorem, the above matrices are
built as follows:  $[\mathbf{D}
]_{ii}=\sum\nolimits_{j\in \mathcal{N}_{i}}\sum%
\nolimits_{l=1}^{L}a_{ij}^{(l)}$ and
$[\mathbf{A}]_{ij}=\sum\nolimits_{l=1}^{L}a_{ij}^{(l)}$. A digraph
is a directed tree if it has $N$ vertices and $N-1$ edges and there
exists a root vertex (i.e., a zero in-degree vertex) with directed
paths to all other vertices. A directed tree is a \emph{spanning}
directed tree of a digraph $\mathscr{G}$ if it has the same vertices
of $\mathscr{G}$.
A digraph is \emph{Strongly} Connected (SC) if, for every pair of nodes $v_i$ and $v_j$ , there exists a
directed path from $v_i$ to $v_j$ and viceversa. A digraph is \emph{Quasi-Strongly} Connected (QSC) if, for every
pair of nodes $v_i$ and $v_j$ , there exists a node $r$ that can reach both $v_i$ and $v_j$ by a directed path.
The fundamental result of this paper is stated in the following.
\begin{theorem}\label{Theorem_delay-linear_stability}
Let $\mathbf{L}$ be the Laplacian matrix associated to the digraph ${\mathscr{G}=}\{{%
\mathscr{V}%
,%
\mathscr{E}}\}$ of system (\ref{linear delayed system}), and let %
$\mathbf{\gamma }=[\gamma
_{1},\ldots ,\gamma _{N}]^{T}$ be the left eigenvector of $\,\mathbf{L}$ corresponding to the zero eigenvalue, i.e., $\mathbf{\gamma }%
^{T}\mathbf{L}=\mathbf{0}_{N}^{T}$. Given system (\ref{linear delayed system}), assume that the
following conditions are satisfied:

\noindent \textbf{a1)} The coupling gain
$K$ and the coefficients $\left\{ c_{i}\right\}$ are positive;

\noindent \textbf{a2)} The propagation delays $\{\tau _{ij}^{(l)}\}$ are finite, the coefficients $\{a_{ij}^{(l)}\}$ are real and the
channel transfer functions $\{H_{rq}(j\omega)\}$ are such that
\begin{equation}
H_{rq}(0)> 0, \quad \forall q,r\neq q,\quad \text{and}\quad \frac{\tsum\nolimits_{q\in \mathcal{N}_{r}}\left\vert
H_{rq}(j\omega )\right\vert }{\tsum\nolimits_{q\in \mathcal{N}_{r}}H_{rq}(0)}%
\leq 1,\quad \forall \omega \in \mathbb{R}, \forall r\neq q;\label{constr_channels}
\end{equation}%

\noindent \textbf{a3)}  The initial conditions are taken in the set of continuously differentiable and bounded
functions mapping the interval $[-\tau ,\ 0]$ to $\mathbb{R} ^{N}$.

Then, system (\ref{linear delayed system}) globally synchronizes,
for \emph{any} set of propagation delays, if and only if the digraph
${\mathscr{G}}$ is QSC. The synchronized state is
\begin{equation}
\alpha^{\star }(\by)=\frac{\dsum_{i=1}^N\gamma_{i}c_{i}g_i(y_i)}
{\dsum_{i=1}^N\gamma_{i}c_{i}+K\dsum_{i=1}^N\gamma_i\dsum_{j\in
\mathcal{N}_{i}}\dsum_{l=1}^L a_{ij}^{(l)}\tau_{ij}^{(l)}},
\label{bias_Theo}
\end{equation}%
where $\gamma _{i}>0$ if and only if node $i$ can reach all the
other nodes of the digraph by a directed path, otherwise $\gamma
_{i}=0$. The convergence is exponential, with asymptotic convergence rate  arbitrarily close
to $r\triangleq -\min_i\{\left \vert\limfunc{Re}\{s_{i}\}\right
\vert: p(s_{i})=0\,\, \text{and}\,\, s_i\neq 0\}$, where $p(s)$ is
the characteristic function associated to system (\ref{linear
delayed system}) (see (\ref{def_characteristic-function}) in the
Appendix).\vspace{-0.2cm}
\end{theorem}\begin{proof}
See the Appendix.\vspace{0.1cm}
\end{proof}

\noindent \textbf{Remark 1 - Robustness against multipath channels:
} Theorem \ref{Theorem_delay-linear_stability} shows that,
differently from classical linear consensus protocols
\cite{Olfati-Saber-Murray-ProcIEEE07, Ren-Beard-Control-Magazine},
the proposed algorithm is robust against propagation delays, since
its convergence condition \emph{is not affected by the delays}.
Moreover, the proposed approach is valid for
\emph{frequency-selective and asymmetric} channels. The only
significant constraint is that the channel coefficients are real and
that, according to \textbf{a2)}, their summation, over each channel,
has to be a positive quantity. This implies a sort of implicit
coherent combination, conceptually similar to the type-based
approach of \cite{Mergen-Tong}, even though the work of
\cite{Mergen-Tong} was aimed at studying the multiple access for
sensor networks with a fusion center, whereas our scheme does not
need a fusion center. The additional constraints on the channel
transfer functions, i.e., ${\tsum\nolimits_{q\in
\mathcal{N}_{r}}\left\vert H_{rq}(j\omega )\right\vert }\le
{\tsum\nolimits_{q\in \mathcal{N}_{r}}H_{rq}(0)}$, for all $\omega
\in \mathbb{R}$, is certainly valid if the channels are low-pass
filters with maximum gain in $\omega =0$. But this is only a
sufficient condition and we will later report some numerical results
showing that if this condition is not satisfied, the method can
still converge. Finally, given the convergence rate $r$, we know for
which class of channels the method is applicable: the channels whose
coherence time is sufficiently greater than $1/r$.

\noindent \textbf{Remark 2 - Effect of network topology:} According
to Theorem \ref{Theorem_delay-linear_stability}, a global consensus
is possible \textit{if and only if} there exists at least one node
(the root node of the spanning directed tree of the digraph) that
can reach all the other nodes by a directed path. If no such a node
exists, the information gathered by each sensor has no way to
propagate through the \emph{whole} network and thus a global
consensus cannot be reached. Moreover, the only nodes contributing
to the final consensus value are the ones having a directed path
linking them to all the other nodes [see (\ref{bias_Theo})]. As a
consequence, the final decision depends on the measurements gathered
by {\it all} the nodes if and only
if the network is strongly connected. When the digraph is not QSC, system (%
\ref{linear delayed system}) may still converge, but it forms separated
clusters of consensus, as proved in \cite{Scutari-Barbarossa-Delay-SPAWC-07,
Barbarossa-Scutari-Magazine}. \medskip

\noindent\textbf{Remark 3 - Unbiased decisions without estimating
the channel parameters: } The closed form expression of the
synchronized state given in (\ref{bias_Theo}) shows a dependence of
the final consensus on the network topology and propagation
parameters. This implies that the final consensus value
(\ref{bias_Theo}), in general, does not coincide with the desired
decision statistics as given in (\ref{f}), except for the trivial
case of flat-fading channels with zero delays and balanced (and thus
strongly connected) digraph. Nevertheless, in the following we
provide a method to get an {\it unbiased} estimate, without having
to get any preliminary estimation of the channel parameters, i.e.
$a_{i, j}^{(l)}, \tau_{i, j}^{(l)}$, incorporating, only in the case
of unbalanced networks, a decentralized estimation of the topology
dependent coefficients $\gamma_i$.

The bias due to the propagation delays and path amplitudes can be
removed using the following two-step algorithm. We let system
(\ref{linear delayed system}) to evolve twice: The first time, the
system evolves according to (\ref{linear delayed system}) and we
denote by $\alpha^\star(\by)$
the synchronized state, as in (\ref{bias_Theo}); the second time, we set $%
g_i(y_i)=1$ in (\ref{linear delayed system}), for all $i$, and the
system is
let to evolve again, denoting the new synchronized state by $\alpha^\star(%
\mbox{\boldmath $1$})$. Taking the ratio
$\alpha^\star(\by)/\alpha^\star(\buno) =(\sum_{i=1}^{N}\gamma
_{i}c_{i} g_{i}(y_i))/(\sum_{i=1}^{N}\gamma
_{i}c_{i}),$ 
we obtain the same consensus value that would have been achieved in
the absence of multipath propagation.

If the network is strongly connected and balanced, $\gamma_i=1,
\forall i$ and then the compensated consensus coincides with the desired value (%
\ref{f}). If the network is unbalanced, the compensated consensus
$\alpha^\star(\by)/\alpha^\star(\buno)$ does not depend on the
multipath coefficients, but it is still biased, with a bias
dependent on $\mathbf{\gamma}$, i.e., on the network topology. This
residual dependence can be eliminated in a decentralized way if
each node is able to estimate its own $\gamma_i$. In fact, in such a case, $\alpha^\star(%
\mbox{\boldmath $y$})/\alpha^\star(\mbox{\boldmath $1$})$ can be
made to coincide with the desired expression in (\ref{f}) by simply
replacing each $c_i$ in (\ref{linear delayed system}) with
$c_i/\gamma_i$, for all $i$ such that $\gamma_i \neq 0$ (suppose that there are $N_r$ of such nodes, w.l.o.g.).
Interestingly, the estimate of each $\gamma_i$ can also be obtained
in a decentralized way, using the following  procedure.
At the beginning, every node sets $g_i(y_i)=1$ and the
network is let to evolve. The final consensus value will be,
in this case $\alpha^\star(\buno)$.  Then, the network is let to evolve $N_r$ times, according to the following protocol. At
step $i$, with $i=1,\ldots N_r,$ node $i$ sets $g_i(y_i)=1$, while all the other nodes set $%
g_k(y_k)=0$ for all $k\neq i$; all nodes are then let to evolve according to (\ref%
{linear delayed system}); let us denote by $\alpha^\star(\bee_i)$
the final consensus value, where $\bee_i$ is the canonical vector
having all zeros, except the $i$-th component, equal to one. 
Each node is now able to take
the ratio $\alpha^\star(\bee_i)/\alpha^\star(\buno)$, which
coincides with the ratio $\tilde{\gamma}_i:=\gamma_i/\sum_k
\gamma_k$. Thus, after $N_r+1$ steps, every node knows its own
(normalized) $\tilde{\gamma}_i$ and it may then use it in the
subsequent run of the consensus algorithm, setting
$c_i=c_/\tilde{\gamma}_i$, to achieve a topology independent
estimate. Observe that, since the eigenvector $\mathbf{\gamma}$ does
not depend on
the observations $\{y_i\}$, the proposed algorithm to estimate $\mathbf{%
\gamma}$ is required to be performed only once every channel
coherence period. In summary, the effects of both delays and channel
coefficients can be eliminated from the final consensus value, even
if at the price of a slight increase of complexity and the need for some
coordination among the nodes.

\section{Numerical Results and Conclusion}
As a numerical example, in the top row of Fig. \ref{Figure-RSCC}, we
report two examples of topologies: the left graph is SC, whereas the
right graph is QSC. For the QSC digraph in the figure, we sketch  its decomposition in Strongly Connected Components (SCC), whose root is denoted by RSCC.\footnote{A SCC of a digraph is a maximal subgraph  which is also SC, meaning that there is no larger SC subgraph containing the nodes of the considered component. A RSCC is a SCC containing all nodes that can reach all the other nodes in the digraph by a directed path \cite[Appendix A]{Scutari-Barbarossa-Delay-SPAWC-07}.} The behavior of the state derivatives versus the
iteration index, in the two cases,  is illustrated on the bottom row
of the figure\footnote{Clearly, the simulations have been performed
on the discretized version of (\ref{linear delayed system}); in such
a case, given the sampling time $T$, there is a maximum value of $K$
guaranteeing the convergence of (\ref{linear delayed system}): $K T$
must be sufficiently smaller than $1/\limfunc{deg}_{\limfunc{in}}^{\max}$, where
$\limfunc{deg}_{\limfunc{in}}^{\max}$ is the maximum in-degree of the graph Laplacian
(cf. \cite[Appendix A]{Scutari-Barbarossa-Delay-SPAWC-07}).}. The edges shown in both graphs show the
active link. Each link is modeled as an FIR filter, modeling the
multipath fading. Each filter has maximum length $L=5$ and the
coefficients have been generated as
$a_{ij}^{(l)}=(A+w(i,j,l))\,e^{-l T/\tau_0}$, $l=0, \ldots, L-1$,
where the constant $A=1$ represents a deterministic component,
whereas $w(i,j,l)$ are i.i.d. random Gaussian variables with zero
mean and standard deviation $\sigma_n=0.5$, modeling the fading.
Observe that, using this setting, some channel coefficients are also
negative. The exponential models the attenuation as a function of
distance and $\tau_0$ represents the delay spread; $T$ is the
sampling time. The delays $\tau_{ij}^{(l)}$ on each link have been
modeled as $\tau_{ij}^{(l)}=d_{ij}/c+(l-1)T$, where $d_{ij}$ is the
distance between nodes $i$ and $j$ and $c$ is the speed of light.
The dimension of the network has been computed in order to make the
maximum delay $\tau_{max}=d_{max}/c$ much larger than the sampling
time $T$. In particular, we chose the parameters so that
$\tau_{max}=30\, T$, in order to test the algorithm under a severe
propagation delay. The constant lines with arrows reported in the
bottom row of Fig. \ref{Figure-RSCC}  represent the theoretical
value, as given by (\ref{bias_Theo}). We can verify that the
simulation curves tend to approach the theoretical values for both
SC and QSC topologies, as predicted by the theory. It is worth
mentioning that, in both cases, we used channels that respect the
condition $H_{rq}(0)>0, \forall r, q$, but do not necessarily
respect the condition ${\tsum\nolimits_{q\in
\mathcal{N}_{r}}\left\vert H_{rq}(j\omega )\right\vert }\le
{\tsum\nolimits_{q\in \mathcal{N}_{r}}H_{rq}(0)}$. Nonetheless,
the simulation results are still in good agreement with our
theoretical findings. There is no contrast with the theory because
the condition ${\tsum\nolimits_{q\in
\mathcal{N}_{r}}\left\vert H_{rq}(j\omega )\right\vert }\le
{\tsum\nolimits_{q\in \mathcal{N}_{r}}H_{rq}(0)}$, $\forall r\neq q,$ is only a sufficient
condition.

\begin{figure}[tbh]
\begin{center}
\includegraphics[height=12 cm]{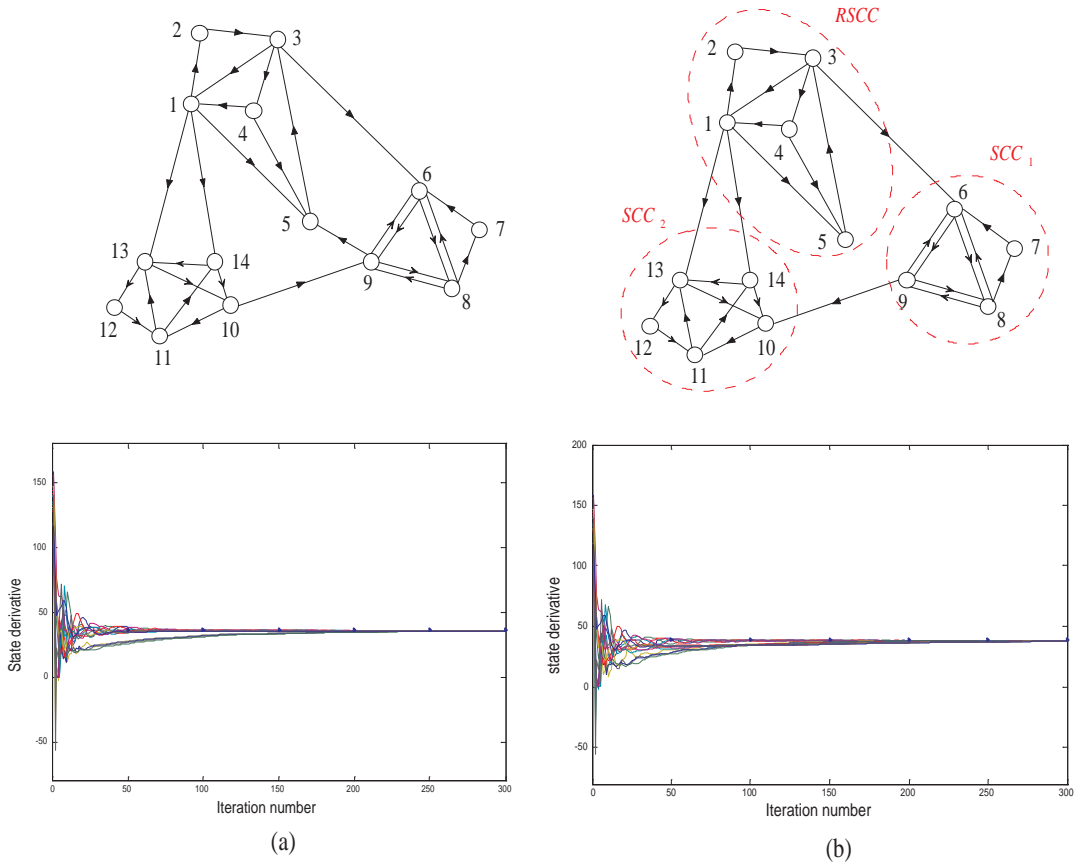}
\end{center}\vspace{-0.9cm}
\caption{{\protect\footnotesize Example of global
consensus in a SC and QSC network.}} \label{Figure-RSCC}\vspace{-0.2cm}
\end{figure}

The estimates reported in Fig. \ref{Figure-RSCC} show a good
agreement between theory and simulation, but the final result does
not coincide with the theoretical optimal value, because of the bias
induced by the multipath coefficients and delays. However, as
suggested at the end of the previous section, it is possible to
remove the bias, without having to estimate neither the channel
amplitudes $a_{i, j}^{(k)}$ nor the delays $\tau_{i, j}^{(k)}$. As
an example of this compensation technique, in Fig.
\ref{Figure-compensated}, we report the running state derivative
$\dot{{x}}_{i}(t; \by)$ (solid lines) and the compensated estimate
$\dot{{x}}_{i}(t; \by)/\dot{{x}}_{i}(t; \buno)$ (dotted line),
together with the theoretical limits (constant lines) achievable
without compensation (triangle marks) and with compensation (circle
marks). The last value coincides with the globally optimal estimate.
The results shown in Fig. \ref{Figure-compensated} have been
achieved with multipath fading channels of length $L=11$, under the
same fading model used in the previous example. Fig.
\ref{Figure-compensated} shows that, as predicted by the theory, the
consensus algorithm with compensation is able to reach the globally
optimal estimate, without the need of estimating the channel
coefficients.\vspace{-0.2cm}

\begin{figure}[tbh]
\begin{center}
\includegraphics[height=6 cm]{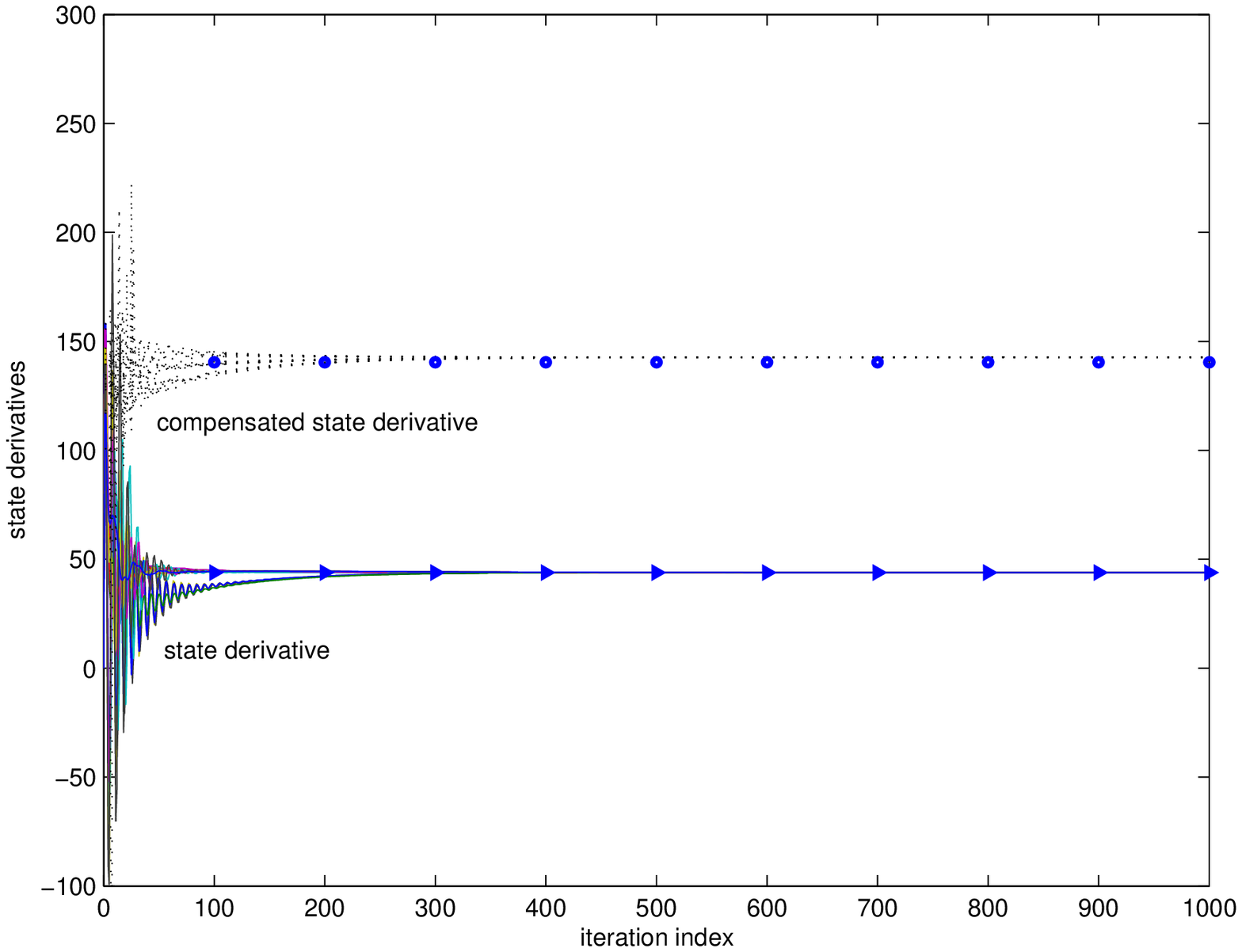}
\end{center}
\vspace{-0.9cm} \caption{{\protect\footnotesize Uncompensated running
estimate (solid lines), i.e. $\dot{{x}}_{i}(t; \by)$ and compensated
running estimate (dotted lines), i.e. $\dot{{x}}_{i}(t;
\by)/\dot{{x}}_{i}(t; \buno)$.}} \label{Figure-compensated}
\end{figure}
In summary, in this work we have derived the conditions allowing a
distributed consensus mechanism to reach globally optimal decision
statistics, in the presence of multipath propagation in the link
between each pair of nodes. The method is valid for real (baseband)
channels and it requires that the summation of the channel
coefficients over each link is strictly positive. Thanks to the
closed form expression derived in this paper, we have also shown how
to get unbiased, globally optimal estimates, without the need to
resolve the signals received from different nodes (thus allowing for
a very simple MAC mechanism), nor to estimate the channel
coefficients. A crucial investigation, motivated from this work, is
the design of the most appropriate radio interface allowing for the
internode interaction enabling the distributed consensus.
\section{Appendix}
Part of the proof of the theorem is based on the same approach we
followed in \cite[Appendix C]{Scutari-Barbarossa-Delay-SPAWC-07}.
Thus, in the following we make use of the general results of
\cite[Appendix C]{Scutari-Barbarossa-Delay-SPAWC-07} and we focus
only on the specific aspects of the model used in this paper.
We start studying the existence of a synchronized state in the form (\ref%
{bias_Theo}). Then, we prove that such a state is also globally
asymptotically stable (cf. Definition 1). Throughout the proof, we
assume that
conditions \textbf{a1)}-\textbf{a4) }are satisfied and that the digraph ${%
\mathscr{G}}$ associated to (\ref{linear delayed system}) is
QSC. In the following, for the sake of notation simplicity,
we drop the dependence of the state function from the observation,
as this dependence does not play any role in our proof.\medskip

\noindent \textbf{Existence of a synchronized state}: The set of
delayed differential equations (\ref{linear delayed system}) admits
a solution in the form
\begin{equation}
x_{i}^{\star }(t)={\alpha }t+x_{i,0}^{\star },\,\quad i=1,\ldots ,N,\vspace{%
-0.1cm}  \label{subclass_of_synch_state}
\end{equation}%
where ${\alpha \in
\mathbb{R}
}$ and $\{x_{i,0}^{\star }\}$ are a set of coefficients that depend
in general on the
system parameters and on the initial conditions, if and only if $%
\{x_{i}^{\star }(t)\}$ satisfies (\ref{linear delayed system}), i.e., if and
only if there exist $\alpha $ and $\{x_{i,0}^{\star }\}$\ such that the
following system of linear equations is feasible:
\begin{equation}
\frac{c_{i}\Delta _{i}(\alpha )}{K}+\sum\nolimits_{j\in \mathcal{N}%
_{i}}\left( \sum\nolimits_{l=1}^{L}a_{ij}^{(l)}\right) \left( x_{j,0}^{\star
}-x_{i,0}^{\star }\right) =0,\vspace{-0.2cm}  \label{system-linear-eq}
\end{equation}%
$\forall i=1,\ldots ,N,$ where\vspace{-0.2cm}
\begin{equation}
\Delta _{i}(\alpha )\triangleq g_{i}(y_{i})-\alpha \left( 1+\frac{K}{c_{i}}%
\sum\nolimits_{j\in \mathcal{N}_{i}}\sum\nolimits_{l=1}^{L}a_{ij}^{(l)}\,%
\tau _{ij}^{(l)}\right) .  \label{Delta_omega}
\end{equation}%
Introducing the weighted Laplacian $\mathbf{L}=\mathbf{L}({\mathscr{G}})$
associated to ${\mathscr{G}}$, system (\ref{system-linear-eq}) can be
equivalently rewritten in vector form as\vspace{-0.1cm}
\begin{equation}
K\mathbf{Lx}_{0}^{\star }=\mathbf{D}_{\mathbf{c}}\boldsymbol{\Delta} \,(\alpha ),\vspace{%
-0.1cm}  \label{sys-linear-eqs}
\end{equation}%
where $\mathbf{x}_{0}^{\star }\triangleq \lbrack x_{1,0}^{\star },\ldots
,x_{N,0}^{\star }]^{T},$ $\mathbf{D}_{\mathbf{c}}\triangleq \limfunc{diag}%
(c_{1},\ldots ,c_{N}),$ and $\boldsymbol{\Delta} \,(\alpha )$ $\triangleq \lbrack \Delta
_{1}(\alpha ),\ldots ,\Delta _{N}(\alpha )]^{T},$ with $\Delta _{i}(\alpha )$
defined in (\ref{Delta_omega}). Observe that, because of \textbf{a2) }and
the quasi-strong connectivity of ${\mathscr{G}}$, the graph Laplacian $%
\mathbf{L}$ has the following properties: i) $\mathrm{rank}(\mathbf{L})=N-1;$
ii) $\mathcal{N}(\mathbf{L})=\mathcal{R}\left( \mathbf{1}_{N}\right) ;$
and iii) $\mathcal{N}(\mathbf{L}^{T})=\mathcal{R}\left( \mathbf{\gamma }%
\right) ,$ where $\mathcal{N}(\cdot)$ and $\mathcal{R}(\cdot)$ denote the (right) null-space and the range space operators, respectively,
 and $\mathbf{\gamma }$ is a left eigenvector of $\mathbf{L}$
corresponding to the (simple) zero eigenvalue of $\mathbf{L}$, i.e., $%
\mathbf{\gamma }^{T}\mathbf{L}=\mathbf{0}^{T}$. It follows from i)-iii)
that, for any \emph{given} $\alpha ,$ (\ref{sys-linear-eqs}) admits a
solution if and only if $\ \mathbf{D}_{\mathbf{c}}\boldsymbol{\Delta} \,(\alpha )\in
\mathcal{R}\left(\mathbf{L}\right)$. Using again properties i)-iii),
we have: $\mathbf{D}_{\mathbf{c}}\boldsymbol{\Delta} (\alpha )\in \mathcal{R}\left (
\mathbf{L}\right) \Leftrightarrow \mathbf{\gamma }^{T}\mathbf{D}_{\mathbf{c}%
}\boldsymbol{\Delta} (\alpha )=0.$ It is easy to check that the value of $\alpha $ that
satisfies the latter condition is $\alpha =\alpha ^{\star },$ with $\alpha
^{\star }$ defined in (\ref{bias_Theo}). Hence, if $\alpha =\alpha ^{\star },
$ the synchronized state in the desired form (\ref{subclass_of_synch_state})
is a solution to (\ref{linear delayed system}), for any given set of $\{\tau
_{ij}^{(l)}\},$ $\{g_{i}\},$ $\{c_{i}\}$, $\{a_{ij}^{(l)}\}$ and $K>0$.

The structure of the left eigenvector $\mathbf{\gamma }$ associated
to the zero eigenvalue of $\mathbf{L}$ as given in the theorem follows from \ \cite[%
Lemma 4]{Scutari-Barbarossa-Delay-SPAWC-07}.

Setting $\alpha =\alpha ^{\star }$, system (\ref{sys-linear-eqs}) admits $%
\infty ^{1}$ solutions, given by $\mathbf{x}_{0}^{\star }=\frac{1}{K}\mathbf{%
\mathbf{L}^{\sharp }D}_{\mathbf{c}}\boldsymbol{\Delta} (\alpha ^{\star })+\mathcal{R}\left(
\mathbf{1}_{N}\right)\triangleq \overline{\mathbf{x}}_{0}+\mathcal{R}\left(\mathbf{%
1}_{N}\right),$ where\vspace{-0.1cm}
\begin{equation}
\overline{\mathbf{x}}_{0}\triangleq \mathbf{\mathbf{L}^{\sharp }D}_{\mathbf{c%
}}\boldsymbol{\Delta} (\alpha ^{\star })/K,\vspace{-0.1cm}  \label{minimum-norm-solution}
\end{equation}%
$\Delta _{i}(\alpha ^{\star })$ is obtained by (\ref{Delta_omega}) setting $%
\alpha =\alpha ^{\star }$ and $\mathbf{\mathbf{L}}^{\sharp }$
is the generalized inverse of the Laplacian $\mathbf{\mathbf{L}}$.\medskip

\noindent \textbf{Global asymptotic stability}: We prove now that the
synchronized state of system (\ref{linear delayed system}) is globally
asymptotically stable. To this end, we use the following intermediate
results.

Let 
$\mathbf{%
\mathbb{C}
}_{+}=\{s\in \mathbf{%
\mathbb{C}
}:\limfunc{Re}\{s\}>0\},$ $\mathbf{%
\mathbb{C}
}_{-}=\{s\in \mathbf{%
\mathbb{C}
}:\limfunc{Re}\{s\}<0\},$ and $\overline{\mathbf{%
\mathbb{C}
}}_{+}$ be the closure of $\mathbf{%
\mathbb{C}
}_{+}$, i.e., $\overline{\mathbf{%
\mathbb{C}
}}_{+}=\{s\in \mathbf{%
\mathbb{C}
}:\limfunc{Re}\{s\}\geq 0\}.$ Denoting by $\mathcal{H}^{n\times m}$ the set
of $n\times m$\ matrices whose entries are analytic\footnote{%
A complex function is said to be analytic (or holomorphic) on a region $%
\mathcal{D\subseteq
\mathbb{C}
}$ if it is complex differentiable at every point in $\mathcal{D}$, i.e.,
for any $z_{0}\in \mathcal{D}$ the function satisfies the Cauchy-Riemann
equations and has continuous first partial derivatives in the neighborhood
of $z_{0}$ (see, e.g., \cite[Theorem 11.2]{Rudin-book}).} and bounded
functions in $\mathbf{%
\mathbb{C}
}_{+},$ let us introduce the degree matrix $%
\boldsymbol{\Delta}\geq \mathbf{0}$ (where \textquotedblleft $\geq$\textquotedblright  has to be intended component-wise)  and the complex matrix $\mathbf{H}(s)\in \mathbf{%
\mathbb{C}
}^{N\times N}\mathbf{,}$ defined respectively as%
\begin{equation}
\boldsymbol{\Delta}%
\triangleq \limfunc{diag}\left( k_{1}\deg \nolimits_{\text{in}%
}(v_{1}),...,k_{N}\deg \nolimits_{\text{in}}(v_{N})\right) ,\qquad \left[
\mathbf{H}(s)\right] _{ij}\triangleq \left\{
\begin{array}{ll}
0, & \text{if }i=j, \\
k_{i}\sum\nolimits_{l=1}^{L}a_{ij}^{(l)}e^{-s\tau _{ij}^{(l)}}, & \text{if }%
i\neq j,%
\end{array}%
\right. \vspace{-0.2cm}  \label{def_H}
\end{equation}%
where $\deg \nolimits_{\text{in}}(v_{i})=$ $\sum_{j\in \mathcal{N}%
_{i}}\sum\nolimits_{l=1}^{L}a_{ij}^{(l)}\geq 0$ is the in-degree of
node $v_i$ and $k_{i}\triangleq K/c_{i}>0.$ Observe that
$\mathbf{H}(s)\in \mathcal{H}^{N\times N}\mathbf{.}$

\begin{lemma}
\label{theorem_stability_of_dynamical_system}Consider the following
linear
functional differential equation:%
\begin{equation}
\begin{array}{l}
\dot{{x}}_{i}(t)=k_{i}\dsum\nolimits_{j\in \mathcal{N}_{i}}\dsum%
\nolimits_{l=1}^{L}a_{ij}^{(l)}\left( x_{j}(t-\tau
_{ij}^{(l)})-x_{i}(t)\right) ,\quad \,\,t>0, \\
x_{i}(\vartheta )=\phi _{i}(\vartheta ),\quad \vartheta \in \lbrack -\tau ,\
0].%
\end{array}%
\quad
\begin{array}{l}
i=1,\ldots ,N,%
\end{array}
\label{system-translated}
\end{equation}%
and assume that the following conditions are satisfied:\vspace{-0.2cm}

\begin{enumerate}
\item[\textbf{b1.}] The initial value functions $\mathbf{\phi }$ are taken
in the set $\ \mathcal{C}^{1}$ of continuously differentiable functions that
are bounded in the norm\footnote{%
We used, without loss of generality, as vector norm in $%
\mathbb{R}
^{N}$ the infinity norm $\left\Vert \cdot \right\Vert _{\infty },$ defined
as $\left\Vert \mathbf{x}\right\Vert _{\infty }\triangleq \max_{i}|x_{i}|.$
Of course, the same conclusions can be obtained using any other norm.} $%
\left\vert \mathbf{\phi }\right\vert _{s}=\sup_{-\tau \leq \vartheta \leq
0}\left\Vert \mathbf{\phi }(\vartheta )\right\Vert _{\infty },\ $and the
solutions $\mathbf{x}(t)$ with initial functions $\mathbf{\phi }$ are
bounded; \vspace{-0.2cm}

\item[\textbf{b2.}] The characteristic equation associated to (\ref%
{system-translated})
\begin{equation}
p(s)\triangleq \det \left( s\mathbf{I}+%
\boldsymbol{\Delta}%
-\mathbf{H}(s)\right) =0,  \label{def_characteristic-function}
\end{equation}%
with $%
\boldsymbol{\Delta}%
$ and $\mathbf{H}(s)$ defined in (\ref{def_H}), has all roots $%
\{s_{r}\}_{r}\in $ $\mathbf{%
\mathbb{C}
}_{-},$ with \ at most one \emph{simple} root at $s=0.$\footnote{%
We assume, w.l.o.g., that the roots $\{s_{r}\}$ are arranged in
nonincreasing order with respect to the real part, i.e., $0=\limfunc{Re}%
\{s_{0}\}>\limfunc{Re}\{s_{1}\}\geq \limfunc{Re}\{s_{2}\}\geq ...$.}
\end{enumerate}

Then, system (\ref{system-translated}) is marginally stable, i.e., $\forall
\mathbf{\phi }\in \mathcal{C}^{1}$ and \ $\limfunc{Re}\{s_{1}\}<c<0,$\ there
exist $t_{1}$ and $\alpha ,$ with $t_{0}<t_{1}<+\infty $ and $0<\alpha
<+\infty ,$ independent of $\mathbf{\phi ,}$ and a vector $\mathbf{x}%
^{\infty },$ with $\left\Vert \mathbf{x}^{\infty }\right\Vert <+\infty ,$
such that
\begin{equation}
\left\Vert \mathbf{x}(t)-\mathbf{x}^{\infty }\right\Vert \leq \alpha
\left\vert \mathbf{\phi }\right\vert _{s}e^{ct},\qquad \forall t>t_{1}.
\label{inequality_convergence_rate}
\end{equation}
\end{lemma}

\begin{proof}
Because of space limitation, we omit the proof that can be obtained
following the same steps of the proof in \cite[Lemma 5]%
{Scutari-Barbarossa-Delay-SPAWC-07}, after observing that system (\ref%
{system-translated}) can be rewritten in the canonical form of \cite[Ch. 6,
Eq. (6.3.2)]{Bellman-Cooke}, \cite[Ch. 3, Eq. (3.1)]{Gu-book}.
\end{proof}

\begin{lemma}[{\protect\cite[Theorem 2.2]{boyd-SH}}]
\label{Lemma-Boyd}Let $\mathbf{H}(s)\in \mathcal{H}^{N\times N}$ and $\rho
\left( \mathbf{H}(s)\right) $ denote the spectral radius of $\mathbf{H(}s%
\mathbf{).}$ Then, $\rho \left( \mathbf{H}(s)\right) $ is a subharmonic%
\footnote{{\scriptsize See, e.g., \cite{boyd-SH}, for the definition of
subharmonic function.}} bounded (above) function on $\overline{\mathbf{%
\mathbb{C}
}}_{+}$.\hfill $\square$
\end{lemma}

We are ready to prove the global asymptotic stability of the synchronized
state of (\ref{linear delayed system}). Applying the following change of
variables: $\Psi _{i}(t)\triangleq x_{i}(t)-\left( \alpha ^{\star }t+%
\overline{x}_{i,0}\right) $, for all $i=1,\ldots ,N,$ 
where $\alpha ^{\star }$ and $\{\overline{x}_{i,0}\}$ are defined in (\ref%
{bias_Theo}) and (\ref{minimum-norm-solution}), respectively, and using (\ref%
{minimum-norm-solution}), the original system (\ref{linear delayed system})
can be equivalently rewritten in terms of $\{\Psi _{i}(t)\}_{i}$ as
\begin{equation}
\dot{\Psi}_{i}(t)=k_{i}\dsum\nolimits_{j\in \mathcal{N}_{i}}\dsum%
\nolimits_{l=1}^{L}a_{ij}^{(l)}\left( \Psi _{j}(t-\tau _{ij}^{(l)})-\Psi
_{i}(t)\right) ,\text{ \ \ }t\geq 0,  \label{system-translated_2}
\end{equation}%
with $\Psi _{i}(\vartheta )=\phi _{i}(\vartheta )\triangleq \widetilde{\phi }%
_{i}(\vartheta )-\alpha ^{\star }\vartheta -\overline{x}_{i,0}$ for $%
\vartheta \in \lbrack -\tau ,\ 0],$ where $\{\widetilde{\phi _{i}}\}$ are
the initial value functions of the original system (\ref{linear delayed
system}).

It follows from (\ref{system-translated_2}) that the synchronized state of
system (\ref{linear delayed system}), as given in (\ref%
{subclass_of_synch_state}) (with $\alpha =\alpha ^{\star }$), is globally
asymptotically stable (according to Definition \ref{Definition_sync-state})
if system (\ref{system-translated_2}) is marginally stable. According to
Lemma \ref{theorem_stability_of_dynamical_system}, the marginal stability of
system (\ref{system-translated_2}) is guaranteed if: \textbf{b1}) the
trajectories $\{\Psi _{i}(t)\}$ are bounded for all $t>0,$ given $\mathbf{%
\phi }\in \mathcal{C}^{1};$ \textbf{b2}) the characteristic equation (\ref%
{def_characteristic-function}) associated to (\ref{system-translated_2}),
has all roots in $\mathbf{%
\mathbb{C}
}_{-},$ with \ at most one \emph{simple} root at $s=0.$

Following the same steps as in \cite[Appendix C]%
{Scutari-Barbarossa-Delay-SPAWC-07}, one can prove that, under \textbf{a1})-%
\textbf{a3}), all the solutions $\{\Psi _{i}(t)\}$ to (\ref%
{system-translated_2}), with initial conditions in $\mathcal{C}^{1},$\ are
uniformly bounded, as required by assumption \textbf{b1}) in Lemma \ref%
{theorem_stability_of_dynamical_system}. Because of space limitation we omit
the details. We study instead the characteristic equation (\ref%
{def_characteristic-function}), and prove that, under \textbf{a1})-\textbf{a3%
}) and the quasi-strong connectivity of the digraph, assumption \textbf{b2})
of Lemma \ref{theorem_stability_of_dynamical_system} is satisfied.

First of all, observe that, since $\boldsymbol{\Delta }-\mathbf{H}(0)=K%
\mathbf{D}_{c}\mathbf{L}$, we have
\begin{equation}
p(0)=\det \left(
\boldsymbol{\Delta}%
-\mathbf{H}(0)\right) =\det \left(K \mathbf{D}_{c}\right) \det \left(
\mathbf{L}\right) =0,  \label{zero_solution}
\end{equation}%
where the last equality in (\ref{zero_solution}) is due to $\mathrm{rank}(%
\mathbf{L})=N-1$. It follows from (\ref{zero_solution}) that $p(s)$ has a
root in $s=0,$ corresponding to the zero eigenvalue of the Laplacian $%
\mathbf{L}$ (recall that $\det \left(K \mathbf{D}_{c}\right) \neq 0$). Since
the digraph is assumed to be QSC, according to \cite[Corollary 2]%
{Scutari-Barbarossa-Delay-SPAWC-07}, such a root is simple. Thus, to
complete the proof, we need to show that $p(s)$ does not have any solution
in $\overline{\mathbf{%
\mathbb{C}
}}_{+}\backslash \{0\},$ i.e.,%
\begin{equation}
\det \left( s\mathbf{I+}%
\boldsymbol{\Delta}%
-\mathbf{H}(s)\right) \neq 0,\quad \forall s\in \overline{\mathbf{%
\mathbb{C}
}}_{+}\backslash \{0\}.  \label{cond_zeros_in_C_r}
\end{equation}%
Since $s\mathbf{I+}%
\boldsymbol{\Delta}%
$ is nonsingular in $\overline{\mathbf{%
\mathbb{C}
}}_{+}\backslash \{0\}$ [recall that, under \textbf{a1)-a2)}, $%
\boldsymbol{\Delta}%
\geq \mathbf{0}$, with at least one positive diagonal entry], (\ref%
{cond_zeros_in_C_r}) is equivalent to
\begin{equation}
\det \left( \mathbf{I-}\left( s\mathbf{I+}%
\boldsymbol{\Delta}%
\right) ^{-1}\mathbf{H}(s)\right) \neq 0,\quad \forall s\in \overline{%
\mathbf{%
\mathbb{C}
}}_{+}\backslash \{0\},
\end{equation}%
which leads to the following sufficient condition for (\ref%
{cond_zeros_in_C_r}):%
\begin{equation}
\rho \left( s\right) \triangleq \rho \left( \left( s\mathbf{I+}%
\boldsymbol{\Delta}%
\right) ^{-1}\mathbf{H}(s)\right) <1,\quad \forall s\in \overline{\mathbf{%
\mathbb{C}
}}_{+}\backslash \{0\}.  \label{cond_spectrum_less_one}
\end{equation}%
Since $\left( s\mathbf{I+}%
\boldsymbol{\Delta}%
\right) ^{-1}\in \mathcal{H}^{N\times N}$ and $\mathbf{H}(s)\in \mathcal{H}%
^{N\times N}$, it follows from Lemma \ref{Lemma-Boyd} that the spectral
radius $\rho \left( s\right) $ in (\ref{cond_spectrum_less_one}) is a
subharmonic function on $\overline{\mathbf{%
\mathbb{C}
}}_{+}.$ As a direct consequence, we have, among all, that $\rho \left(
s\right) $ is a continuous bounded function on $\overline{\mathbf{%
\mathbb{C}
}}_{+}$ and satisfies the \emph{maximum modulus principle }(see, e.g., \cite%
{boyd-SH} and references therein): $\rho \left( s\right) $ achieves its
\emph{global }maximum only on the boundary of $\overline{\mathbf{%
\mathbb{C}
}}_{+}.$ Since $\rho \left( s\right) $ is strictly proper in $\overline{%
\mathbf{%
\mathbb{C}
}}_{+},$ i.e., $\rho \left( s\right) \rightarrow 0$ as $\left\vert
s\right\vert \rightarrow +\infty $ while keeping $s\in \overline{\mathbf{%
\mathbb{C}
}}_{+}$, it follows that $\sup_{s\in \mathbf{%
\mathbb{C}
}_{+}}\rho \left( s\right) <\sup_{\omega \in
\mathbb{R}
}\rho \left( j\omega \right) .$ According to the latter inequality,
condition (\ref{cond_spectrum_less_one}) is satisfied if
\begin{equation}
\rho \left( j\omega \right) =\rho \left( \left( j\omega \mathbf{I+}%
\boldsymbol{\Delta}%
\right) ^{-1}\mathbf{H}(j\omega )\right) <1,\quad \forall \omega \in
\mathbb{R}
\backslash \{0\}.
\end{equation}%
Denoting by $\left\Vert \mathbf{A}\right\Vert _{\infty }\triangleq
\max_{r}\sum\nolimits_{q}|\left[ \mathbf{A}\right] _{rq}|$ the \emph{maximum
row sum matrix norm }and using $\rho \left( \mathbf{A}\right) \leq
\left\Vert \mathbf{A}\right\Vert _{\infty }$ $\forall \mathbf{A\in
\mathbb{C}
}^{N\times M}$ \cite{Horn85}, we have

\begin{eqnarray*}
\rho \left( j\omega \right) & \leq
&\left\Vert \left( j\omega \mathbf{I+}%
\boldsymbol{\Delta}%
\right) ^{-1}\mathbf{H}(j\omega )\right\Vert _{\infty }\!\!=\!\max_{r}\frac{%
\tsum\nolimits_{q\in \mathcal{N}_{r}}\left\vert k_{r}H_{rq}(j\omega
)\right\vert }{\left\vert j\omega +k_{r}\tsum\nolimits_{q\in \mathcal{N}%
_{r}}H_{rq}(0)\right\vert }  \leq \max_{r}\frac{\tsum\nolimits_{q\in
\mathcal{N}_{r}}\left\vert
H_{rq}(j\omega )\right\vert }{\tsum\nolimits_{q\in \mathcal{N}_{r}}H_{rq}(0)}%
\leq 1,
\end{eqnarray*}
where in the last inequality we used (\ref{constr_channels}) [see
assumption \textbf{a2)}]. Since in the second inequality, the
equality is reached if and only if $\omega =0,$ we have $\rho \left(
j\omega \right) <1$ for all $\omega \neq 0,$ which guarantees that
condition (\ref{cond_spectrum_less_one}) is satisfied. Hence,
according to Lemma \ref{theorem_stability_of_dynamical_system},
given any set of initial conditions $\{\phi _{i}\}$ satisfying
\textbf{a4)}, the
trajectories $\mathbf{\Psi }(t)\rightarrow \mathbf{\Psi }^{\infty }$ as $%
t\rightarrow +\infty ,$ with exponential rate arbitrarily close to $%
r\triangleq -\{\min_{i}|\limfunc{Re}\{s_{i}\}|:p(s_{i})=0\,\,\text{and}%
\,\,s_{i}\neq 0\}$, where $p(s)$ is defined in (\ref%
{def_characteristic-function}) and $\mathbf{\Psi }^{\infty }\in \mathcal{R}(\mathbf{1}_{N})$ (because of $\mathbf{L\Psi }^{\infty }=\mathbf{0}$ and $%
\mathrm{rank}(\mathbf{L})=N-1$\textbf{)}. In other words, system (\ref%
{system-translated_2}) exponentially reaches the consensus on the state.
\medskip

\noindent \textbf{Necessity}: The necessity of quasi-strong connectivity of the digraph for the
network to reach a \emph{global} consensus  can be proved as in
\cite[Appendix C.2]{Scutari-Barbarossa-Delay-SPAWC-07}  by showing that, if the digraph associated to (\ref{System-model}) is not QSC,
different clusters of nodes synchronize on different values \cite[Corollary 1]{Scutari-Barbarossa-Delay-SPAWC-07}. This local synchronization is in contrast with the
definition of (global) synchronization, as given in Definition 1. Hence, if the overall network has to
synchronize, the digraph associated to the system must be QSC.

\def\baselinestretch{0.90}


\normalsize
\begin{thebibliography}{99}
\scriptsize
\bibitem{Olfati-Saber-Murray-ProcIEEE07} { R. Olfati-Saber, J. A.
Fax, R. M. Murray, ``Consensus and Cooperation in Networked Multi-agent
Systems,'' in \textit{Proc. of the IEEE}, vol. 95, no. 1, pp. 215-233, Jan.
2007. }

\bibitem{Ren-Beard-Control-Magazine} { W. Ren, R. W. Beard, and
E. M. Atkins, ``Information Consensus in Multivehicle Cooperative Control:
Collective Group Behavior Through Local Interaction,'' \textit{IEEE Control
Systems Mag.}, vol. 27, no. 2, pp. 71-82, April 2007. }

\bibitem{Olfati-Saber} { R. Olfati-Saber and R.M. Murray,
\textquotedblleft Consensus Problems in Networks of Agents with Switching
Topology and Time-Delays," \textit{IEEE Trans. on Automatic Control,} vol.
49, pp. 1520-1533, Sep., 2004. }

\bibitem{Strogatz} { M.G. Earl, and S.H. Strogatz,
\textquotedblleft Synchronization in Oscillator Networks with Delayed
Coupling: A Stability Criterion,\textquotedblright \textit{Physical Rev. E},
Vol. 67, pp. 1-4, 2003. }

\bibitem{Papachristodoulou-CDC06} A. Papachristodoulou and A.
Jadbabaie, \textquotedblleft Synchronization in Oscillator Networks:
Switching Topologies and Presence of Nonhomogeneous
delays,\textquotedblright in Proc of the \textit{IEEE ECC-CDC '05}, Dec. 2005.

\bibitem{Lee-Spong-06} { D.J. Lee and M.W. Spong,
\textquotedblleft Agreement With Non-uniform Information Delays," in Proc.
of the \textit{ACC '06}, June 2006. }

\bibitem{Tsitsiklis-Bertsekas-Athans} { J. N. Tsitsiklis, D. P.
Bertsekas, M. Athans, ``Distributed Asynchronous Deterministic and
Stochastic Gradient Optimization Algorithms,'' \textit{IEEE Trans. on
Automatic Control}, pp. 803--812, Sep. 1986. }

\bibitem{Tsitsiklis-Bertsekas-book} { D. P Bertsekas and J.N.
Tsitsiklis, \textit{Parallel and Distributed Computation: Numerical Methods}%
, Athena Scientific, 1989. }

\bibitem{Blondel-Tsitsiklis} { V. D. Blondel, J. M. Hendrickx, A.
Olshevsky, and J. N. Tsitsiklis, ``Convergence in Multiagent Coordination,
Consensus, and Flocking," in Proc. of the \textit{IEEE CDC-ECC'05}, Dec. 2005. }

\bibitem{boyd-gosh-prabhakar-shah}
S. Boyd, A. Ghosh, B. Prabhakar, D. Shah, \textquotedblleft
Randomized gossip algorithms,\textquotedblright in \textit{IEEE
Trans. on Information Theory}, Vol. 52, pp. 2508--2530, June 2006.

\bibitem{Scutari-Barbarossa-Delay-SPAWC-07} { G. Scutari, S.
Barbarossa, and L. Pescosolido, \textquotedblleft Distributed
Decision Through Self-Synchronizing Sensor Networks in the Presence
of Propagation Delays and Asymmetric Channels,\textquotedblright\, in
\textit{IEEE Trans. on Sign. Proc.}, vol. 56, no. 4, pp.1667--1684, April 2008.}

\bibitem{Barbarossa-Scutari-Journal} { S. Barbarossa, and G.
Scutari, \textquotedblleft Decentralized Maximum Likelihood Estimation for
Sensor Networks Composed of Self-synchronizing Locally Coupled
Oscillators,\textquotedblright, \textit{IEEE Trans. on Sign.
Proc.}, vol. 55,  no. 7,  pp. 3456--3470, July 2007.}

\bibitem{Barbarossa-Scutari-Magazine} { S. Barbarossa and G.
Scutari, ``Bio-inspired Sensor Network Design: Distributed Decision Through
Self-synchronization,'' \textit{IEEE Sign. Proc. Magazine}, vol. 24,
no. 3, pp. 26--35, May 2007. }

{
\bibitem{Rudin-book} W. Rudin, \textit{Real and Complex Analysis}, McGraw-Hill, International Student Ed., 1970.
}

\bibitem{boyd-SH}S. Boyd and S. A. Desoer, \textquotedblleft
Subharmonic Functions and Performance Bounds on Linear Time-Invariant
Feedback Systems,\textquotedblright in \textit{IMA Jour. of Math.
Control $\&$ Inf.}, Vol. 2, pp. 153-170, 1985.


\bibitem{Bellman-Cooke} R. Bellman and K.L. Cooke, {\it Differential-Difference Equations}, New York Academic Press, 1963.

\bibitem{Gu-book}K. Gu, V.L. Kharitonov, J. Chen, \emph{%
Stability of Tyme-Delay Systems, }Control Engineering Series,\emph{\ }%
Birkhauser, 2002.

\bibitem{Horn85} R. A. Horn and C. R.
Johnson, \textit{Matrix Analysis}, Cambridge Univ. Press, 1985.


\bibitem{Mergen-Tong}
G. Mergen, L. Tong, \textquotedblleft Type-based Estimation over
Multiaccess Channels,\textquotedblright in \textit{IEEE Trans. on Signal Proc.}, Vol. 54, pp. 613--626, Febr. 2006.


\bibitem{Pescosolido-Barbarossa-icassp2008}
L. Pescosolido and S. Barbarossa \textquotedblleft Distributed Decision in Sensor Networks based on
Local Coupling through Pulse Position Modulated Signals,\textquotedblright in Proc. of IEEE  \textit{ICASSP 2008}, March 30 - April 4, 2008, Las Vegas, NV, (USA).

\bibitem{Pescosolido-radar} L. Pescosolido, S. Barbarossa, and G. Scutari, \textquotedblleft Radar Sensor Networks with Distributed Detection  Capabilities,\textquotedblright  in Proc.  of  \textit{IEEE Radar Conference 2008}, Sheraton Golf Parco dei Medici,  May 26-30, 2008, Rome, Italy.

\bibitem{Barbarossa-energy} S. Barbarossa, G. Scutari, A. Swami, “Achieving Consensus in Self-Organizing Wireless Sensor Networks: The Impact of Network Topology on Energy Consumption,”\,  in Proc. of \textit{IEEE ICASSP 2007}, April 15-20, Honolulu, Hawaii (USA).


\end{thebibliography}
\end{document}